\documentclass[11pt, a4paper]{article}
\usepackage{fullpage}

\usepackage[latin1]{inputenc}

\usepackage{amsmath}
\usepackage{amssymb,amsthm,graphicx}
\usepackage[dvips]{epsfig}
\usepackage{amsfonts}
\usepackage{xcolor}
\usepackage{hyperref}
\usepackage{enumerate}
\usepackage{authblk}

\newtheorem{theorem}{Theorem}

\newtheorem{corollary}[theorem]{Corollary}
\newtheorem{observation}[theorem]{Observation}

\theoremstyle{remark}

\theoremstyle{definition}

\newcommand{\N}{\mathbb{N}}
\newcommand{\Z}{\mathbb{Z}}

\definecolor{darkgreen}{rgb}{0,.5,0}

\newcounter{sideremark}

\title{VEST is $W[2]$-hard
\thanks{
The author acknowledges support by the project ``Grant Schemes at CU'' (reg. no. CZ.02.2.69/0.0/0.0/19\_073/0016935)).
}
}

\author{Michael Skotnica}

\affil{Department of Applied Mathematics, Faculty of Mathematics and Physics,
Charles~University, Prague, Czech Republic}

\date{}

\begin{document}

\maketitle

\begin{abstract}
In this short note, we show that the problem of VEST is $W[2]$-hard for parameter $k$.
This strengthens a result of Matou\v{s}ek, who showed $W[1]$-hardness of that problem.
The consequence of this result is that computing the $k$-th homotopy group of a $d$-dimensional space
for $d > 3$ is $W[2]$-hard for parameter $k$.
\end{abstract}

\section{Introduction}
The homotopy groups $\pi_k$, for $k = 1,2,\ldots$ are important invariants of topological spaces.
The most intuitive of them is the group $\pi_1$ which is often called fundamental group.

Many topological spaces can be described by finite structures, e.g. by abstract simplicial complexes.
Such structure can be used as an input for a computer and therefore,
it is natural to ask how hard is to compute these homotopy groups
of a given topological space represented by an abstract simplicial complex.

Novikov in 1955 (see \cite{Novikov55}) and independently Boone in 1959 (see \cite{Boone59})
showed undecidability of the word problem for groups.
Their result also implies undecidability of computing the fundamental group.
(Even determining whether the fundamental group of a given topological space is trivial
is undecidable.)

On the other hand, it is known that for greater $k$, the corresponding
homotopy group $\pi_k$ is a finitely generated abelian group which is always isomorphic to a group of the form
\begin{align*}
  \Z^n \oplus \Z_{p_1} \oplus \Z_{p_2} \oplus \cdots \oplus\Z_{p_m}
\end{align*}
where $p_1, \ldots, p_m$ are powers of prime numbers.%
\footnote{Note that $\Z^n$ is a direct sum of $n$ coppies of $\Z$ while $\Z_{p_i}$ is a finite cyclic group of order $p_i$.}
An algorithm for computing $\pi_k$, where $k > 1$, was first introduced by Brown in 1957 (see \cite{Brown57}).

In 1989, Annick (see \cite{Anick89}) proved that computing rank of $\pi_k$, that is the number of direct summands isomorphic to $\Z$ (represented by $n$ in the expression above) is $\#P$-hard for 4-dimensional 1-connected spaces.
Another computational problem called VEST, which we define below, was used in Annick's proof as an intermediate step.
Briefly said, $\#P$-hardness of the problem of VEST implies $\#P$-hardnes of computing rank of $\pi_k$.

\paragraph{Vector evaluated after a sequence of transformations (VEST).}
The input of this problem defined by Anick in \cite{Anick89}
is a vector $\textbf{v} \in \mathbb{Q}^d$, a list of $(T_1,\ldots, T_m)$
of rational $d \times d$ matrices and a rational matrix $S \in \mathbb{Q}^{h \times d}$
for $d,m,h \in \N$.

For an instance of a VEST let \emph{$M$-sequence} be a sequence of integers $M_1, M_2, M_3, \ldots$, where
\begin{align*}
  M_k := |\{(i_1,\ldots, i_k); ST_{i_k}\cdots T_{i_1} \mathbf{v} = \mathbf{0}\}|.
\end{align*}
Given an instance of a VEST and $k \in \N$, the goal is to compute $M_k$.
Note, that instead of rational setting we can assume integral setting.

From an instance of a VEST, it is possible to construct a corresponding algebraic structure called \emph{$123H$-algebra}
in polynomial time whose \emph{Tor-sequence} is equal to the $M$-squence of the original instance of a VEST.
This is stated in \cite[Theorem 3.4]{Anick89} and it follows from \cite[Theorem 1.3]{Anick85} and \cite[Theorem 7.6]{Anick87}.

Given a presentation of a $123H$-algebra, one can construct a corresponding 4-dimensional simplicial complex in polynomial time
whose sequence of ranks $(\textup{rk }\pi_2, \textup{rk }\pi_3, \ldots)$ is related to the Tor-sequence of the $123H$-algebra.
In particular, it is possible to compute that Tor-sequence from the sequence of ranks using an $FPT$ algorithm. (Which is defined in the next paragraph). This follows from \cite{Roos79} and \cite{Cadek14_2}. To sum up, a hardness of computing $M_k$ of a VEST implies a hardness of computing $\pi_k$.

\paragraph{Parameterized complexity and $W$ hierarchy}
It is also possible to look at the problem of computing $\pi_k$ from the viewpoint of parameterized complexity which
classifies decision computational problems with respect to a given parameter(s).
For instance, one can ask if there exists an independent set of size $k$ in a given graph,
where $k$ is the parameter. 

In our case, the number $k$ of the homotopy group $\pi_k$ plays the role of such parameter.
Since we assume only decision problems, we only ask whether the rank of $\pi_k(X)$ of a space $X$ is nonzero (or equal to a particular number).
In 2014 \v{C}adek et al. (see \cite{Cadek14}) proved that this problem is in $XP$ in the parameter $k$.
In other words, there is an algorithm solving this problem in time $cn^{f(k)}$, where $c$ is a constant, $n$ is the size of input
and $f(k)$ is a computable function of the parameter $k$.

A lower bound for the complexity from the parameterized viewpoint was obtained by Matou\v{s}ek in 2013 (see \cite{Matousek13}).
He proved that computing $M_k$ of a VEST is $W[1]$-hard. This also implies $W[1]$-hardness for the original problem of computing higher homotopy groups $\pi_k$ (for 4-dimensional 1-connected spaces) for parameter $k$.
Matou\v{s}ek's proof also works as a proof for $\#P$-hardness and it is shorter and much more easier than the original proof of Annick in \cite{Anick89}.

The class $W[1]$ is a member of the following $W$ hierarchy, which we briefly define.

\begin{align*}
  FPT \subseteq W[1] \subseteq W[2] \subseteq \cdots \subseteq W[P] \subseteq XP
\end{align*}

We have already defined the class $XP$ above. The class $FPT$ consists of decision problems solvable in time $f(k)n^{O(1)}$,
where $f(k)$ is a computable function of the parameter $k$  and $n$ is the size of input.
It is only known that $FPT \subsetneq XP$ (see \cite{Flum04}). The class $W[1]$ then consists of all problems
which can be reduced by an $FPT$ algorithm to a booliean circuit of a constant depth with AND, OR  and NOT gates such that there is at most 1 gate of higher input size than 2 on each path from the input gate to the final output gate (this number of larger gates is called \emph{weft}) such that the parameter $k$ from the original problem is translated to setting $g(k)$ input gates to TRUE. See Figure~\ref{Fig:clique}.
It is strongly believed that $FPT \subsetneq W[1]$. \emph{Therefore, one cannot expect existence of an algorithm solving a $W[1]$-hard problem in time $f(k)n^{c}$
where $f(k)$ is a computable function of $k$ and $c$ is a constant.}

\begin{figure}
  \centering
  \includegraphics[scale=.862]{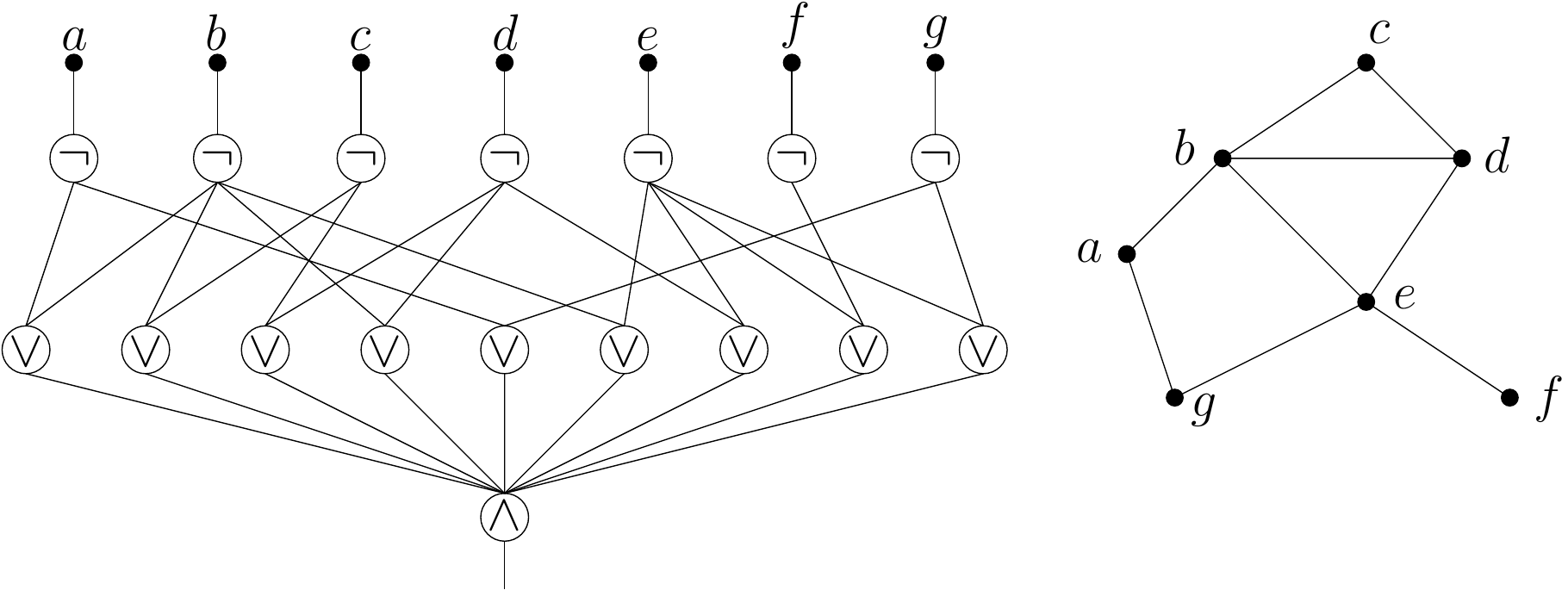}
  \caption{A boolien circuit solving the problem of existence of an independent set of size $k$ in the graph on the left.
  There is an independent set of size $k$ in the graph if and only if the boolean circuit outputs TRUE
  for an input consisting of exactly $k$ true values.
  }
  \label{Fig:clique}
\end{figure}

The class $W[i]$ consists of problems $FPT$-reducible to a boolean circuit of a constant depth and weft at most $i$.

The class $W[P]$ can be defined as a class of problems which can be solved by non-deterministic Turing machine
which can make at most $O(g(k)\log n)$ non-deterministic choices and which works in time $f(k)n^{O(1)}$.
See~\cite{Flum04}.

According to this definition, it is easy to see that the problem of VEST is in $W[P]$.

\begin{observation}
Computing $M_k$ of a VEST for parameter $k$ is in $W[P]$.
\end{observation}
\begin{proof}
  Let $n$ be the size of the input and $m$ the number of the matrices in the collection.
  In~particular, $m \leq n$.
  
  We can guess which $k$ matrices we choose from the collection. Each matrix can be represented
  by an integer $\leq m$ which can be described by $\log m$ bits.
  Therefore, we need at most $k\log m \leq k\log n$ non-deterministic choices.
  
  Then, we need to multiply $k+1$ matrices together with 1 vector. This  can be easily done in
  time $(k+2)n^3$.
\end{proof}

In this note, we strengthen the result of Matou\v{s}ek and show that the problem of VEST is $W[2]$-hard.
Our proof is even simpler than the proof of $W[1]$-hardness.

\begin{theorem}\label{thm:w2harndess}
  Computing $M_k$ of a VEST is $W[2]$-hard for parameter $k$ .
\end{theorem}

Theorem~\ref{thm:w2harndess} together with the result of Anick (see \cite{Anick89}) implies the following.

\begin{corollary}
  Computing $k$-th homotopy groups of $d$-dimensional space
  for $d > 3$ is $W[2]$-hard in the parameter $k$.
\end{corollary}

\section{The proof} 
Note that the current complexity of the problem of VEST is a self-contained problem.
Our reduction will use only 0,1 matrices and the initial vector $\mathbf{v}$.
Moreover, each matrix will have at most one 1 in each line.
Therefore, such construction also shows $W[2]$-hardness of a VEST for $\Z_2$ setting.

\paragraph{$W[2]$-complete problem.}
Our reduction is from well-know problem of existence of a dominating set of size $k$
which is known to be $W[2]$-complete and which we define in this paragraph. See~\cite{Flum04}.

For a graph $G(V,E)$ and its vertex $v \in V$
let $N[v]$ denote the closed neighborhood of a vertex $v$.
That is, $N[v] := \{u \in V; \{u,v\} \in E\} \cup \{v\}$.

A \emph{dominating set} of a graph $G(V,E)$ is a set $U \subseteq V$
such that for each $v$ there is $u \in U$ such that $v \in N[u]$.

\begin{proof}[Proof of Theorem~\ref{thm:w2harndess}]
  We show an FPT reduction from the problem of existence of a dominating set of size $k$
  to a VEST.
  
  Let $G(V,E)$ be a graph and let $n = |V|$. We start with a description of our vector space.
  It is of dimension $3n + 1$. For each $u \in V$ we have 3 dimensions $u_1, u_2, u_3$.  
 Then there is one extra dimension $c$. In the beginning, the corresponding coordinates of the vector $\textbf{v}$ are set as follows:  $u_1 = 1$ and $u_2 = u_3 = 0$. The coordinate corresponding
 to $c$ will be set to 1 during the whole computations. The described coordinates will simulate a data structure during the computation which will correspond to matrix multiplication.
  
For each $u \in V$ we create a matrix $M_u$ as follows.
This matrix nullifies the coordinate $w_1$ for each $w \in N[u]$
which corresponds to a domination of vertices in $N[u]$ by the vertex $u$.
The matrix $M_u$ also set $u_2$ to $u_3$ and $u_3$ to 1.
This can be done by the coordinate $c$.
See Figure~\ref{fig:no_repetition}.
\begin{figure}
  \centering
  $\begin{pmatrix}
    0 & 1 & 0\\
    0 & 0 & 1\\
    0 & 0 & 1
  \end{pmatrix}$
  \caption{The submatrix realizing the procedure which assures that no matrix can repeat.
    The first line corresponds to the coordinate $u_2$, the second to $u_3$ and the third to $c$.}\label{fig:no_repetition}
\end{figure}
It is similar to the procedure from \cite{Matousek13} and it assures that each matrix can be chosen only once, but
comparing to that procedure used by Matou\v{s}ek it also works in $\Z_2$ setting;
note that in the beginning or after one multiplication by matrix $M_u$ the coordinate $u_2 = 0$
while after two or more multiplications it is $1$.

The matrix $S$ then chooses the coordinates $u_1$ and $u_2$ for each vector $u$.
The coordinate $u_1 = 1$ if and only if corresponding vertex is not dominated.
As it was discussed in the previous paragraph $u_2 = 0$
if and only if the corresponding matrix was not chosen or was chosen only once.

Therefore, $M_k = k!D_k$ where $D_k$ is the number of dominating sets of size $k$ of graph $G$.

Note that the reduction is FPT.
Indeed, we do not use the parameter $k$ during it
and it is polynomial in the size of input.
\end{proof}

\bibliographystyle{alpha}

\begin{thebibliography}{{\v{C}}KM{\etalchar{+}}14b}

\bibitem[Ani85]{Anick85}
David~J. Anick.
\newblock {D}iophantine equations, {H}ilbert series, and undecidable spaces.
\newblock {\em Annals of {M}athematics}, 122:87--112, 1985.

\bibitem[Ani87]{Anick87}
David~J. Anick.
\newblock {G}eneric algebras and {C}{W} complexes.
\newblock {\em {A}lgebraic topology and algebraic {K}-theory}, pages 247--321,
  1987.

\bibitem[Ani89]{Anick89}
David~J. Anick.
\newblock The computation of rational homotopy groups is \#$\wp$-hard.
  {C}omputers in geometry and topology, {P}roc. {C}onf., {C}hicago/{I}ll. 1986,
  {L}ect. {N}otes {P}ure {A}ppl. {M}ath. 114.
\newblock pages 1--56, 1989.

\bibitem[Boo59]{Boone59}
William~W. Boone.
\newblock The word problem.
\newblock {\em Annals of mathematics}, 70:207--265, 1959.

\bibitem[Bro57]{Brown57}
Edgar~H. Brown.
\newblock {F}inite computability of {P}ostnikov complexes.
\newblock {\em {A}nnals of {M}athematics}, 65:1, 1957.

\bibitem[{\v{C}}KM{\etalchar{+}}14a]{Cadek14_2}
Martin {\v{C}}adek, Marek Kr{\v{c}}{\'a}l, Ji{\v{r}}{\'\i} Matou{\v{s}}ek,
  Luk{\'a}{\v{s}} Vok{\v{r}}{\'\i}nek, and Uli Wagner.
\newblock Extendability of continuous maps is undecidable.
\newblock {\em Discrete \& Computational Geometry}, 51(1):24--66, 2014.

\bibitem[{\v{C}}KM{\etalchar{+}}14b]{Cadek14}
Martin {\v{C}}adek, Marek Kr\v{c}\'{a}l, Ji\v{r}\'{\i} Matou\v{s}ek,
  Luk\'{a}\v{s} Vok\v{r}\'{\i}nek, and Uli Wagner.
\newblock Polynomial-time computation of homotopy groups and {P}ostnikov
  systems in fixed dimension.
\newblock {\em {S}{I}{A}{M} {J}ournal on {C}omputing}, 43(5):1728--1780, 2014.

\bibitem[FG04]{Flum04}
J\"{o}rg Flum and Martin Groge.
\newblock {\em Parameterized Complexity Theory}.
\newblock Springer, 2004.

\bibitem[Mat13]{Matousek13}
Ji\v{r}\'{i} Matou\v{s}ek.
\newblock Computing higher homotopy groups is {W}[1]-hard.
\newblock {\em arXiv preprint arXiv:1304.7705}, 2013.

\bibitem[Nov55]{Novikov55}
Pyotr~S. Novikov.
\newblock On the algorithmic unsolvability of the word problem in group theory.
\newblock {\em Trudy {M}at. {I}nst. {S}teklov}, 44:1--143, 1955.
\newblock (in Russian).

\bibitem[Roo79]{Roos79}
Jan-Erik Roos.
\newblock Relations between the poincar{\'e}-betti series of loop spaces and of
  local rings.
\newblock In {\em S{\'e}minaire d'Alg{\`e}bre Paul Dubreil}, pages 285--322.
  Springer, 1979.

\end{thebibliography}
\newcommand{\etalchar}[1]{$^{#1}$}

\end{document}